\documentclass[a4paper, USenglish, cleveref, autoref, thm-restate]{lipics-v2021}

\hideLIPIcs

\title{Emit As You Go: Enumerating Edges of a Spanning Tree}

\author{Katrin Casel}{Humboldt Universität zu Berlin, Germany}{Katrin.Casel@hu-berlin.de}{https://orcid.org/0000-0001-6146-8684}{}
\author{Stefan Neubert}{Hasso Plattner Institute, University of Potsdam, Germany}{Stefan.Neubert@hpi.de}{https://orcid.org/0000-0001-9148-6592}{}

\authorrunning{K. Casel and S. Neubert}

\Copyright{Katrin Casel and Stefan Neubert} 


\category{}
\relatedversion{}

\supplement{The source code for our experiments is published as follows (including algorithm implementations and generators for random input instances):}
\supplementdetails[subcategory={Source Code}, swhid={swh:1:rev:f69e607d59c7058c57c2a217db8b509f8a5f14e2}]{Software}{https://github.com/Steditor/Enumeration-of-Solution-Parts}

\nolinenumbers

\usepackage{csquotes}

\usepackage{tikz}
\usetikzlibrary{positioning, calc, matrix, arrows.meta}
\tikzstyle{vertex} = [draw, circle, minimum width=1.5em]
\tikzstyle{directed} = [draw, ->, >=Latex]

\usepackage[linesnumbered, ruled, vlined, algo2e]{algorithm2e}
\SetFuncSty{textsc}
\SetDataSty{textit}

\usepackage{booktabs}
\usepackage{tabularx}
\usepackage{makecell}
\usepackage{thmtools}

\usepackage{pifont}
\newcommand{\cmark}{\ding{51}}
\newcommand{\xmark}{\ding{55}}

\usepackage{pgfplots}
\pgfplotsset{compat=1.5}
\usepgfplotslibrary{fillbetween}

\newcommand{\maxdeg}{\Delta}
\newcommand{\maxoutdeg}{\Delta^+}
\newcommand{\maxindeg}{\Delta^-}
\newcommand{\avgdeg}{\overline{\Delta}}
\newcommand{\avgoutdeg}{\overline{\Delta^+}}
\newcommand{\avgindeg}{\overline{\Delta^-}}
\newcommand{\outdeg}{\deg^+}
\newcommand{\indeg}{\deg^-}
\newcommand{\totaltime}{\mathcal{T}}
\newcommand{\prefixtime}{\mathcal{T}_k}
\newcommand{\totalspace}{\mathcal{S}}

\definecolor{plotcolor}{HTML}{007f6f}

\begin{document}

\maketitle

\begin{abstract}
	Classically, planning tasks are studied as a two-step process: plan creation and plan execution.
	In situations where plan creation is slow (for example, due to expensive information access or complex constraints), a natural speed-up tactic is interleaving planning and execution.
	We implement such an approach with an enumeration algorithm that, after little preprocessing time, outputs parts of a plan one by one with little delay in-between consecutive outputs.
	As concrete planning task, we consider efficient connectivity in a network formalized as the minimum spanning tree problem in all four standard variants: (un)weighted (un)directed graphs.
	Solution parts to be emitted one by one for this concrete task are the individual edges that form the final tree.

	We show with algorithmic upper bounds and matching unconditional adversary lower bounds that efficient enumeration is possible for three of four problem variants;
	specifically for undirected unweighted graphs (delay in the order of the average degree), as well as graphs with either weights (delay in the order of the maximum degree and the average runtime per emitted edge of a total-time algorithm) or directions (delay in the order of the maximum degree).
	For graphs with both weighted and directed edges, we show that no meaningful enumeration is possible.

	Finally, with experiments on random undirected unweighted graphs, we show that the theoretical advantage of little preprocessing and delay carries over to practice.
\end{abstract}

\section{Introduction}
\label{sec:introduction}

In many complex planning settings, such as path finding for robots in a storage facility, the combined time of planning and execution determines the overall efficiency.
A natural way of optimizing this is to start execution before fully finishing planning \cite{zhangPlanningExecutionMultiAgent2024, neubertIncrementalOrderingScheduling2024}.
This approach generalizes to all kinds of multi-step processes, in which the output of an early step is needed to start working on a later one.
In fact, with large input instances it can even be beneficial to employ an algorithm with worse total-time, if that algorithm produces the solution in the form of many \textit{solution parts} and (while still running) emits them for following steps to process \cite{lindnerHiLiveRealtimeMapping2017}.

But how much speedup can actually be gained by this approach?
We tackle this question for one of the most fundamental problems on networks: cost-efficient connectivity.
Be it networks for communication, electricity, transport or water supply -- both for planning problems and for the analysis of existing infrastructure it is a central task to connect every node to all other nodes by a Minimum Spanning Tree (MST).
Besides being useful as output on their own, MSTs are also used as input for many complex algorithmic tasks, e.\,g. for Coverage Path Planning \cite{tanComprehensiveReviewCoverage2021}, or in data analytics for graph-based clustering \cite{zahnGraphTheoreticalMethodsDetecting1971,gagolewskiClusteringMinimumSpanning2024}, and for image segmentation \cite{felzenszwalbEfficientGraphBasedImage2004}.

Naturally, there is an extensive history of research on the complexity of computing a \textit{(minimum) spanning tree} for (un)directed and (un)weighted graphs -- that is if by \textit{complexity} one refers to the \textit{total-time complexity} of the problem:
How many computation steps does it take in total to produce a complete solution as output?
With the above motivation of \textit{starting next steps early}, recently several core algorithmic problems have been analyzed using terminology from \textit{enumeration complexity}~\cite{strozeckiEnumerationComplexity2019}:
How much \textit{preprocessing} time does it take at most to produce a first part of the solution, and what is the worst case \textit{delay} in-between two consecutive solution parts that are emitted?
(Note that in contrast to classical enumeration this question does not ask to enumerate all solutions to an input instance, but to enumerate all parts of a single solution.)
Ideally, an algorithm does not need to spend any time on preprocessing but is able to immediately emit solution parts with little delay.
On the other end of the spectrum are problems for which the time to compute even the first few parts of a solution is in the same order as the time required to solve the problem completely.
The goal therefore is to understand possible tradeoffs between preprocessing time and delay.

\begin{figure}[t]
	\centering
	\begin{tikzpicture}
		\begin{loglogaxis}[
			width=215pt,
			height=110pt,
			xlabel={Number of vertices},
			ylabel={CPU time in ns},
			ymin=100,
			ymax=10e11,
			legend entries={
				e-first,
				e-total,
				prim-total,
			},
			legend style={at={(0.5,1)},anchor=south,draw=none},
			legend columns=3,
		]
			\addplot[plotcolor, mark=diamond*, mark size=1.4] table
				[x=size,y=first_output_avg, col sep=comma] {./experiment_data/aggregated_0.25.enum-prim.csv};
			\addplot[plotcolor, mark=square*, mark size=1.4] table
				[x=size,y=total_time_avg, col sep=comma] {./experiment_data/aggregated_0.25.enum-prim.csv};

			\addplot[black, mark=square, mark size=1.4] table
			[x=size,y=total_time_avg, col sep=comma] {./experiment_data/aggregated_0.25.total-prim.csv};

			\addplot[draw=none, name path=e-first-min] table
				[x=size,y=first_output_lower_quartile, col sep=comma] {./experiment_data/aggregated_0.25.enum-prim.csv};
			\addplot[draw=none, name path=e-first-max] table
				[x=size,y=first_output_upper_quartile, col sep=comma] {./experiment_data/aggregated_0.25.enum-prim.csv};

			\addplot[draw=none, name path=e-total-min] table
				[x=size,y=total_time_lower_quartile, col sep=comma] {./experiment_data/aggregated_0.25.enum-prim.csv};
			\addplot[draw=none, name path=e-total-max] table
				[x=size,y=total_time_upper_quartile, col sep=comma] {./experiment_data/aggregated_0.25.enum-prim.csv};

			\addplot[draw=none, name path=p-total-min] table
				[x=size,y=total_time_lower_quartile, col sep=comma] {./experiment_data/aggregated_0.25.total-prim.csv};
			\addplot[draw=none, name path=p-total-max] table
				[x=size,y=total_time_upper_quartile, col sep=comma] {./experiment_data/aggregated_0.25.total-prim.csv};
		\end{loglogaxis}
	\end{tikzpicture}
	\caption{Time to first output and total time of MST enumeration (e) compared to the runtime of Prim's MST algorithm.}
	\label{fig:intro-experiments}
\end{figure}
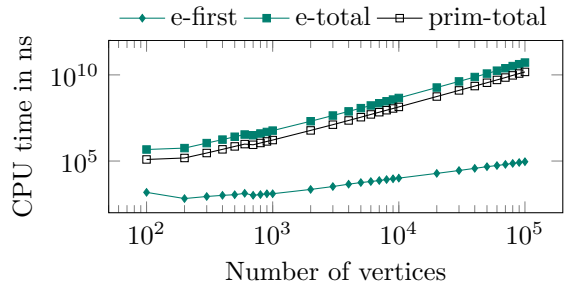

Among the problems studied with this perspective are incremental sorting~\cite{paredesOptimalIncrementalSorting2006,navarroSortingHeapsMinimum2010}, shuffling~\cite{carmeliAnsweringUnionsConjunctive2022}, computation of all pairs shortest distances in a graph~\cite{caselShortestDistancesEnumeration2024} and finding a topological ordering in a directed acyclic graph for solving scheduling problems~\cite{neubertIncrementalOrderingScheduling2024}.
Interestingly, one of the applications of incremental sorting is speeding up total-time MST algorithms; improved also as \emph{Filter Kruskal} in~\cite{osipovFilterKruskalMinimumSpanning2009}.
While these algorithms \textit{employ} the incremental idea for sub problems, we go one step further and discuss MST \textit{itself} in terms of enumerating solution parts.

\paragraph*{Our Contribution}
We present a family of efficient algorithms that enumerate edges of a spanning tree for a given input graph.
For MSTs of weighted undirected graphs, we experimentally confirm in \Cref{sec:experiments} that our enumeration approach leads to a smaller time to a first emitted solution part and small delay in practice, with only a constant overhead in the total time (see \Cref{fig:intro-experiments}).

We complement the algorithmic results with lower bounds to almost completely characterize the complexity of enumerating the edges of a spanning tree for a connected graph with $n$ vertices and $m$ edges.
\Cref{tab:results} summarizes our results for all typical edge types:

\begin{itemize}
	\item \textbf{Undirected unweighted} edges allow for a delay in the order of the average degree~$\avgdeg$; we derive a matching lower bound from the best possible total-time algorithm (\Cref{sec:uust}).
	\item For \textbf{undirected weighted} edges we show a lower bound in the order of the maximum degree $\maxdeg$ on either preprocessing time or delay.
	We also present algorithms that, for small $\maxdeg$ or with $O(n)$ preprocessing time, homogeneously spread out the computation time $\totaltime$ of a total-time algorithm for minimum spanning trees and emit the $n-1$ solution parts with delay in $O(\frac{\totaltime}{n})$ (\Cref{sec:umst}).
	\item Given a \textbf{directed unweighted} graph along with the root for a spanning tree, computing a shortest path tree with delay in the order of the maximum out-degree $\maxoutdeg$ \cite{caselShortestDistancesEnumeration2024} solves the problem with optimal delay, as we show with a matching lower bound (\Cref{sec:dust})
	\item Without given root or when the graph has \textbf{directed weighted} edges, no algorithm can emit any part of a solution after less than $\Omega(m)$ time.
	For unweighted edges this already matches the total time needed to solve the instance.
	But also the small remaining gap to existing total time algorithms for the weighted case (see below for related work) rules out any meaningful enumeration (\Cref{sec:dmst}).
\end{itemize}

\begin{table*}[t]
	\caption{%
		Our results for different problem variants.
		We abbreviate \textit{preprocessing} with prep.
		The symbol $\totaltime$ represents the runtime of an (optimal) MST total-time algorithm.
		An \xmark{} marks variants where lower bounds rule out any meaningful enumeration.
	}
	\label{tab:results}
	\centering
	\begin{tabularx}{\textwidth}{lXX}
		\toprule
		Problem Variant & Lower Bounds & Upper Bounds \\ \midrule
		Undirected Unweighted & \hyperref[thm:uust-lower]{$\Omega(m)$ prep. or $\Omega(\avgdeg)$ delay} & \hyperref[thm:uust-upper]{$O(\avgdeg)$ delay w/o prep.} \\
		Undirected Weighted & \makecell[tl]{\hyperref[thm:mst-maxdeg-lower]{$\Omega(\maxdeg)$ prep. or $\Omega(\maxdeg)$ delay}\\ \hyperref[cor:mst-algo-lower]{$\Omega(\totaltime)$ prep. or $\Omega(\frac{\totaltime}{n})$ delay}} & \makecell[tl]{\hyperref[thm:mst-maxdeg-upper]{$O(\max(\maxdeg,\frac{\totaltime}{n}))$ delay w/o prep.}\\ \hyperref[thm:mst-algo-upper]{$O(\frac{\totaltime}{n})$ delay w/ $O(n)$ prep.}} \\
		Directed Unweighted & & \\
		\qquad w/ given root & \hyperref[thm:dste-with-r-lower]{$\Omega((\maxoutdeg)^2)$ prep. or $\Omega(\maxoutdeg)$ delay} & \hyperref[cor:dste-with-r-upper]{$O(\maxoutdeg)$ delay w/o prep.} \\
		\qquad w/o given root & \hyperref[thm:dste-lower]{$\Omega(m)$ prep. or delay} & {\xmark} \hyperref[thm:dst-upper]{\textit{$O(m)$ total time}} \\
		Directed Weighted & \hyperref[thm:mdste-lower]{$\Omega(m)$ prep. or delay} & {\xmark} \textit{$O(m + n\log(n))$ total time \cite{gabowEfficientAlgorithmsFinding1986}} \\ \bottomrule
	\end{tabularx}
\end{table*}

We build upon the model and techniques introduced in \cite{caselShortestDistancesEnumeration2024}.
We formalize these in \Cref{sec:preliminaries} where we also describe the general form of the algorithms we develop.

Our algorithms usually start with a targeted search for a set of easy-to-find solution parts.
As this is then extended to a full solution by tailored versions of existing algorithms for finding spanning trees, it makes sense to summarize the previous work in this field:

A spanning tree for an unweighted graph can be computed in $O(m + n)$, e.\,g. with a depth first search from a given root \cite{cormenIntroductionAlgorithms2022}; we show in \Cref{thm:dst-upper} how to cover the case without given root.

There are numerous approaches to solve the \textit{minimum weight} spanning tree problem for \textit{undirected} edge-weighted graphs (and the equivalent problem of maximizing the tree's total edge weight).
While we refer to \cite{grahamHistoryMinimumSpanning1985,nesetrilFewRemarksHistory1997,ericksonAlgorithms2019} for a comprehensive account of the many (re)discoveries, we want to mention the approaches for the comparison-based computation model coined by Borůvka (\cite{boruvkaJistemProblemuMinimalnim1926}; $O(m\log(n))$), Jarník, Prim, and Dijkstra (\cite{jarnikJistemProblemuMinimalnim1930,primShortestConnectionNetworks1957,dijkstraNoteTwoProblems1959}; $O(m+n\log(n))$ with Fibonacci heaps \cite{fredmanFibonacciHeapsTheir1987}), Kruskal (\cite{kruskalShortestSpanningSubtree1956}; $O(m\log(n))$), Chazelle (\cite{chazelleMinimumSpanningTree2000}; $O(m\alpha(m,n))$), and Pettie and Ramachandran (\cite{pettieOptimalMinimumSpanning2002}; unknown, but optimal complexity).
For the RAM model, Fredman and Willard \cite{fredmanTransdichotomousAlgorithmsMinimum1994} developed a linear-time algorithm.

Previous work on \textit{directed} minimum spanning trees (sometimes called \enquote{branchings} or \enquote{arborescences}), centers around the so-called Edmonds' algorithm, its (re)discoveries and a series of implementation improvements with a complexity of $O(m + n\log(n))$ in the comparison-based model and $O(m\log\log(n))$ in the RAM model \cite{chuShortestArborescenceDirected1965,edmondsOptimumBranchings1967,bockAlgorithmConstructMinimum1971,karpSimpleDerivationEdmonds1971,tarjanFindingOptimumBranchings1977,cameriniNoteFindingOptimum1979,gabowEfficientAlgorithmsFinding1986,mendelsonMeldingPriorityQueues2006}.

While the majority of the research has concentrated on the total-time model, there are several other takes on how to provide input data, expect output data, and measure complexity for computing spanning trees.
\textit{Dynamic} algorithms have to maintain a (minimum) spanning tree for a graph undergoing edge insertions and deletions \cite{holmPolylogarithmicDeterministicFullydynamic2001}.
Such algorithms have to make changes to the existing output, meaning they cannot be used for the problem of enumerating spanning tree edges of a static graph.
The same applies to the \textit{reconfiguration} model, where an algorithm has to transform one spanning tree into another while preserving the spanning tree properties \cite{itoComplexityReconfigurationProblems2011}.
Closer to our setting is a formalization as online problem (e.\,g.~\cite{bergOnlineMinimumSpanning2023,megowPowerRecourseOnline2012}).
There, same as for enumeration, algorithms have to produce the final solution piece by piece.
In the online setting, however, also the input arrives in pieces which only allows for approximate solutions or requires \textit{recourse actions}.
In contrast to this, enumeration algorithms have access to the complete input and are able to produce an exact solution.

\section{Preliminaries and Techniques}
\label{sec:preliminaries}

We start by introducing graph notation, defining spanning tree problems and the enumeration framework, and explaining the techniques we use for proving upper and lower bounds.

\subsection{Graph Notation and Data Structure}
\label{subsec:notation}

Let $G = (V, E)$ be a (directed or undirected) graph with vertices $V = \{1, \ldots, n\}$ and $|E| = m$ edges.
We will in this work assume that input graphs are connected, as spanning trees only exist for connected graphs.
Note that this implies $m \geq n-1$ and allows us to shorten $O(m+n)$ to $O(m)$ in asymptotic complexity analysis.
Besides the unweighted case, we also consider graphs with edge weights given as function $w\colon E \to \mathbb{R}$.

We denote the degree of a vertex $v$ by $\deg(v)$, the average degree in the graph by $\avgdeg = \frac{1}{n} \cdot \sum_{v\in V} \deg(v)$, and the maximum degree by~$\maxdeg$.
In case $G$ is directed, we write $\outdeg(v)$ and $\indeg(v)$ for the out- and in-degree of $v$, $\maxoutdeg$ and $\maxindeg$ for the maximum out- and in-degree, and $\avgoutdeg$ and $\avgindeg$ for the average out- and in-degree.

Graphs are always given in the form of adjacency lists.
For a graph $G$ and a vertex $v$ we denote the adjacency list of $v$ as $G[v]$.
We assume that the length of each adjacency list (thus: the degree of each vertex) can be queried in constant time, and that edge weights (if present) are stored alongside the respective entry in the adjacency lists.

Our analyses assume the word RAM model with word size in $\Omega(\log(n))$; enough to store a reference to any vertex or edge in a constant number of cells that can be read and written in constant time.
For edge weights we only assume that they can be compared in constant time.
Note that we use existing algorithms for computing minimum spanning trees as black-boxes; your choice of algorithm there might imply additional requirements on the machine model.

\subsection{Spanning Trees}

A \textit{forest} of an undirected graph $G = (V,E)$ is a subgraph $F = (V, E')$ of $G$ with $E' \subseteq E$ that is acyclic.
A \textit{spanning tree} (ST) of $G$ is a connected forest $T = (V, E')$ of $G$ with $|E'| = n-1$.
With additional edge weights $w$ a \textit{minimum spanning tree} (MST) of $G$ is a spanning tree that minimizes $\sum_{e \in E'} w(e)$.

An \textit{out-branching} of a directed graph $G = (V,E)$ is an acyclic subgraph $B = (V, E')$ of $G$ with $E' \subseteq E$  in which each vertex $v$ has $\indeg(v) \leq 1$.
If an out-branching is weakly connected (thus: it is connected if one ignores edge directions), we have $|E'| = n - 1$ and thus there is exactly one vertex $r$ with $\indeg(r) = 0$.
We call such an out-branching $T$ a \textit{directed spanning tree} (DST) of $G$ rooted in $r$.
With additional edge weights $w$ a \textit{minimum directed spanning tree} (MDST) of $G$ is a directed spanning tree that minimizes $\sum_{e \in E'} w(e)$.

\subsection{Enumeration of Solution Parts}

Given an (un)directed (un)weighted input graph $G$, an \textit{enumeration algorithm} has to, for some (minimum) (directed) spanning tree $T$ of $G$, emit each edge of $T$ exactly once.
We analyze the time complexity of such an enumeration algorithm in terms of its worst case \textit{delay}, that is the maximum time the algorithm needs to emit the respective next edge, and additional \textit{preprocessing time} the algorithm can spend before emitting the first edge.
The tradeoff between little preprocessing time and small delay is the core of our analyses.

For sufficiently small preprocessing time, an enumeration algorithm might not be able to initialize complex data structures before having to emit the first solution parts.
We follow the approach presented in \cite{caselShortestDistancesEnumeration2024} to allow for \textit{lazy-initialized} memory that can be reserved within constant time and specify for all our results how much lazy-initialized memory an algorithm needs along with its total space complexity.

\subsection{Amortized Analysis for Upper Bounds}
\label{subsec:upper-techniques}

When designing and analyzing an enumeration algorithm it is important to separate \emph{computing} solution parts from \emph{emitting} them.
While an algorithm might fix parts of the solution after irregular time intervals, it can possibly shrink the maximum output delay by holding back some solution parts to emit later.
This transformation was first described in \cite{goldbergEfficientAlgorithmsListing1991} and comes with the downside of having to store the held-back data.
However, with the spanning tree problems at hand, an algorithm has to store at most $n-1$ edge references, which does not increase its asymptotic space complexity.

Our algorithms hold back solution parts for later emission by storing them in a linked list called \textit{solution queue}.
The key aspect of an algorithm's description is how this queue is filled with solution parts; complemented by an analysis showing that the solution queue never runs empty when a solution part has to be emitted.
Note that we expect an algorithm to actively emit held-back solution parts.
While we do not explicitly write out the instructions on when to dequeue and emit a solution part from the solution queue, we do have to explain how an algorithm is able to compute a lower bound on the aspired preprocessing and delay in order to be able to emit the next solution part after a suitable number of computation steps.

We analyze the time complexity of our algorithms with the banker's view on amortization~\cite{tarjanAmortizedComputationalComplexity1985}:
An algorithm holds a number of \textit{credits}.
Assume we claim that an algorithm solves a spanning tree enumeration problem with preprocessing in $O(p)$ and delay in $O(d)$.
Initially, for some implementation-specific constant $c$, the algorithm starts with $c \cdot (p+d)$ credits.
For every solution part it emits after $\Theta(d') \subseteq O(d)$ computation steps, it receives additional $c \cdot d'$ credits.
Each credit pays for a fixed constant number of computation steps, and every step has to be paid for.
To prove our claim we thus have to show that the algorithm's account balance cannot become negative.
With this technique in mind, our algorithms and their proofs are roughly structured in three phases as follows:
\begin{description}
	\item[1. Credit Accumulation] Fix simple initial parts of the solution within little computation time and compute a lower bound on the aspired delay. Emit only enough solution parts to pay for this initial effort and hold back the rest to save credits for the next step.
	\item[2. Extension] Extend the set of initial solution parts to a complete solution.
		Pay for the required computation with held-back solution parts from the credit accumulation phase.
	\item[3. Output Finalization] Emit the remaining solution parts from phases~1 and~2 without repeating a solution part.
\end{description}

\subsection{Proof Techniques for Lower Bounds}
\label{subsec:lower-techniques}

For proving lower bounds on the time complexity of enumerating a spanning tree we analyze how evenly the computational effort for a set of solution parts could possibly be spread over corresponding delays.
These kinds of lower bounds come in two different flavors:

Any algorithm enumerating all $n-1$ edges of a spanning tree with preprocessing in $O(p)$ and delay in $O(d)$ yields a total-time algorithm with time complexity in $O(p + n \cdot d)$.
Assume it takes $\Omega(\totaltime)$ to produce all $n-1$ edges of a spanning tree (either by an unconditional lower bound or by comparing to an algorithm with runtime in $\Theta(\totaltime)$ as benchmark).
Then every enumeration algorithm has to have preprocessing in $\Omega(\totaltime)$ or delay in $\Omega(\frac{\totaltime}{n})$ or it would beat the total-time bound.

For a more fine-grained analysis we observe that we do not need to have a lower bound on the total time to compute all solution parts, but can also make use of a lower bound of $\Omega(\prefixtime)$ on the required computational effort for finding any first $k$ solution edges.
Similarly to before, this yields a lower bound of either $\Omega(\prefixtime)$ on preprocessing or $\Omega(\frac{\prefixtime}{k})$ on delay.
For this kind of lower bound we construct adversary arguments:
We formalize an algorithm's read access to the input data as queries to an adversary.
This adversary can freely choose how to answer the queries as long as all answers combined are consistent with at least one actual possible input graph.
The adversary uses this freedom to hide essential information on the instance from the algorithm (usually inside densely connected subgraphs) and thereby forces the algorithm to use a lot of queries to fix the first set of $k$ solution parts.

The interface of our adversary mimics providing the input in the form of adjacency lists.
Usually, an algorithm has constant time random access to the degree of any vertex and to the head of its adjacency list.
It can then iterate through the list to determine the neighbors of this vertex.
Accordingly, our adversaries allow for the following queries for any vertex $v$:
\begin{description}
	\item[degree query] Returns the out-degree of $v$.
	\item[neighbor query] Returns the next out-adjacency of $v$; along with its edge weight if present.
\end{description}

\section{Theoretical Results}

We now apply the introduced techniques to computing (minimum) spanning trees for graphs with and without directions or edge weights.

\subsection{Undirected Unweighted Spanning Trees}
\label{sec:uust}

For connected input graphs with neither edge directions nor edge weights we derive an enumeration lower bound from an unconditional lower bound on the total time required for computing a spanning tree.
Recall for the following theorem that $m$ is the number of edges and $\avgdeg$ is the average degree of the graph.

\begin{theorem}
	\label{thm:uust-lower}
	No algorithm can compute a spanning tree of an unweighted, undirected, connected graph in $o(m)$.
	No algorithm can enumerate the edges of such a spanning tree with preprocessing in $o(m)$ and delay in $o(\avgdeg)$.
\end{theorem}
\begin{proof}
	Consider a graph with two almost-clique components with sizes $\lceil\frac{n}{2}\rceil$ and $\lfloor\frac{n}{2}\rfloor$ connected by a bridge of two edges as shown in \Cref{fig:st-total-lower}.
	In both components all vertices are connected to all others except for two vertices which do not share a common edge but are endpoints to the bridge connecting to the other component.

	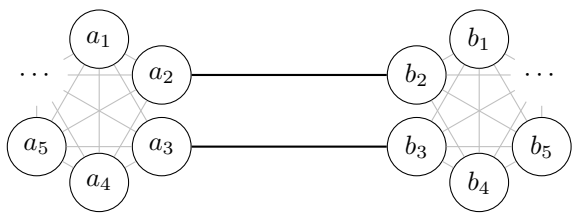
\begin{figure}[t]
		\centering
		\begin{tikzpicture}[vertex/.append style = {minimum width=2.2em}]
			\foreach \v in {1,...,5}{
				\node[vertex] (a\v) at (150-60*\v:0.95cm) {$a_\v$};
			}
			\node[vertex,draw=none] (a6) at (150:0.95cm) {$\cdots$};
			\foreach \u in {1,...,5}{
				\pgfmathtruncatemacro{\unext}{\u+1}
				\foreach \v in {\unext,...,6}{
					\ifthenelse{\u = 2 \and \v = 3}{}{
						\draw[lightgray] (a\u) -- (a\v);
					}
				}
			}
			\begin{scope}[xshift=5cm]
				\foreach \v in {1,...,5}{
					\node[vertex] (b\v) at (30+60*\v:0.95cm) {$b_\v$};
				}
				\node[vertex,draw=none] (b6) at (30:0.95cm) {$\cdots$};
				\foreach \u in {1,...,5}{
					\pgfmathtruncatemacro{\unext}{\u+1}
					\foreach \v in {\unext,...,6}{
						\ifthenelse{\u = 2 \and \v = 3}{}{
							\draw[lightgray] (b\u) -- (b\v);
						}
					}
				}
			\end{scope}

			\draw[thick]
				(a2) edge (b2)
				(a3) edge (b3)
			;
		\end{tikzpicture}
		\caption{The adversarial graph for \Cref{thm:uust-lower}. An algorithm can be trapped in the cliques before being able to detect one of the two connecting edges; one of which must be part of the tree.}
		\label{fig:st-total-lower}
	\end{figure}

	We use the adversarial setup described in \Cref{subsec:lower-techniques}.
	Observe that within each component all vertices have the same degree, thus an algorithm cannot identify the bridge endpoints with degree queries.
	Also, as long as the algorithm has not queried the complete neighborhood of at least all but $2$ vertices in one component, the adversary is free to, for neighbor queries, always return an edge from the queried vertex to an adjacent vertex in the same component.

	Assume, an algorithm were to return a spanning tree of the input graph after $o(m)$ queries to the adversary.
	Clearly, this spanning tree has to connect the two components, so let, w.\,l.\,o.\,g., edge $\{a_2, b_2\}$ be part of the output.
	But as the algorithm has not yet queried at least $\lfloor \frac{n}{2} \rfloor - 2$ complete neighborhoods of size at least $\lfloor \frac{n}{2} \rfloor - 1$, it has not seen $\{a_2, b_2\}$ and the adversary can swap the bridge endpoints in one clique such that the edge does not actually exist (and the algorithm should have returned $\{a_2, b_3\}$ or $\{a_3, b_2\}$ instead).

	Thus, no algorithm can compute a complete spanning tree for a graph in $o(m)$.
	This implies a lower bound of $\Omega(m)$ on the total time of an enumeration algorithm for the same problem.
	It follows that no algorithm can enumerate the $n-1$ solution edges of a spanning tree with preprocessing in $o(m)$ and delay in $o(\frac{m}{n-1}) = o(\avgdeg)$.
\end{proof}

A simple depth first search (DFS) started on any vertex produces in $O(m)$ a spanning tree in the form of a DFS-tree, matching the lower bound in the total-time setting.
The lower bound is tight for the enumeration variant as well, as we now design an algorithm that enumerates a spanning tree without preprocessing and with delay in $O(\avgdeg)$.
Roughly, the three phases of this algorithm work as follows.
In the credit accumulation phase, it collects $\frac{n}{2} \leq k \leq n - 1$ edges with constant delay that form a forest.
These edges give enough head start to, in the extension phase, run a modified version of Prim's MST algorithm in $O(m)$ that selects additional $n - 1 - k$ edges to extend the forest to a spanning tree.
The output finalization phase emits the remaining edges from the first two phases.

\subparagraph*{Phase 1: Credit Accumulation.}
All undirected edges selected in the first phase as part of the final spanning tree are stored in a graph data structure $F$ in adjacency lists representation; initially $F$ is a graph with $n$ isolated vertices.
For each vertex $u$ with empty $F[u]$, the algorithm selects the first edge $\{u,v\}$ in the adjacency list $G[u]$ and adds the edge to $F$ (both $u$ to $F[v]$ and $v$ to $F[u]$ are added).
Simultaneously, the algorithm computes the average degree $\avgdeg$ by summing up all vertex degrees and dividing by $n$.

Recall from \Cref{subsec:upper-techniques} that, in addition to storing edges in $F$, our algorithm appends all edges selected for emission to a \textit{solution queue} from which the algorithm has to emit edges with delay in $O(\avgdeg)$.
The first $\frac{n}{4}$ selected edges from the credit accumulation phase are emitted with constant delay (as the algorithm needs time to compute $\avgdeg$).
All following edges are emitted with delay in $\Theta(\avgdeg)$.

\subparagraph*{Phase 2: Extension.}
The algorithm runs a modified version of Prim's algorithm \cite{jarnikJistemProblemuMinimalnim1930,primShortestConnectionNetworks1957,dijkstraNoteTwoProblems1959}:
Starting from an arbitrary vertex, an MST is constructed by repeatedly selecting from a priority queue an edge of minimum weight that connects a previously unconnected vertex to the growing spanning tree.
In our unweighted case we prioritize pre-selected edges over all other edges to make sure all edges from phase~1 are picked by the algorithm.
We give a complete pseudo-code description of our modifications in \Cref{alg:extension-prim}.

The algorithm stores the tree in the form of predecessor attributes in an array~$T$ of length~$n$, where a vertex $v$ connected to the tree via an edge $\{u,v\}$ has $T[v] = u$.
Additionally, the algorithm maintains two queues of directed edges that are to be used next to connect vertices to the growing tree; one list for pre-selected edges and one list for all edges, where the former is prioritized in edge selection over the later.
Pre-selected edges are added to both queues as filtering all adjacent edges according to whether they are part of $F$ or not is too expensive.
Skipping unneeded edges later on the other hand is cheap by checking $T$ in line~\ref{ln:prim-visit}.

\begin{algorithm2e}
	\KwIn{undirected graph $G = (V,E)$, corresponding forest $F$ from phase~1}
	\KwOut{spanning tree encoded as predecessor links in array $T$ of length $|V|$}
	\SetKwFunction{KwEnqueue}{Enqueue}
	\SetKwFunction{KwDequeue}{Dequeue}
	\SetKwData{KwPreSelected}{preSelectedEdges}
	\SetKwData{KwAllEdges}{allEdges}

	\KwPreSelected = $\emptyset$, \KwAllEdges = $\emptyset$\;
	\lForEach{$v \in V$}{%
		$T[v] = \textsc{nil}$%
	}

	\KwEnqueue{\KwAllEdges, $(1, 1)$}\tcp*[l]{fictitious start edge}

	\While{$\KwPreSelected \neq \emptyset$ or $\KwAllEdges \neq \emptyset$}{ \label{ln:prim-loop}
		\lIf{$\KwPreSelected \neq \emptyset$}{\label{ln:prim-dequeue}%
			$(u,v) = \KwDequeue{\KwPreSelected}$%
		}
		\lElse{%
			$(u,v) = \KwDequeue{\KwAllEdges}$%
		}
		\If{$T[v] == \textsc{nil}$}{ \label{ln:prim-visit}
			$T[v] = u$\;
			\lForEach{$w \in F[v]$}{%
				\KwEnqueue{\KwPreSelected, $(v, w)$}%
			} \label{ln:prim-expand}
			\lForEach{$w \in G[v]$}{%
				\KwEnqueue{\KwAllEdges, $(v, w)$}%
			}
		}
	}

	$T[1] = \textsc{nil}$\tcp*[l]{remove fictitious start edge}

	\KwRet{T}\;

	\caption{Extension-Prim}
	\label{alg:extension-prim}
\end{algorithm2e}

\subparagraph*{Phase 3: Output Finalization.}
The algorithm iterates over all edges in the forest $F$ from the first phase, sets the corresponding entries in the tree $T$ from the second phase to \textsc{nil} and adds for all remaining non-\textsc{nil}-entries $T[v] = u$ the edge $\{u,v\}$ to the solution queue.

\begin{theorem}
	\label{thm:uust-upper}
	Enumerating the edges of a spanning tree of an unweighted, undirected, connected graph $G$ can be done without preprocessing, with delay in $O(\avgdeg)$, with $\Theta(n)$ lazy-initialized memory and with space complexity in $\Theta(n)$.
\end{theorem}
\begin{proof}
	We first prove that the presented algorithm produces a valid spanning tree and emits each edge exactly once, before analyzing the algorithm's time and space complexity.

	\subparagraph*{Correctness}

	In the credit accumulation phase each vertex is adjacent to at least one selected edge, thus $k \geq \frac{n}{2}$ edges are selected.
	If edge $\{u,v\}$ was selected from $G[u]$, we call $u$ the \textit{origin} of this edge.
	If there were a cycle of length $c$ in $F$, each of the $c$ vertices in the cycle would have to be the origin of one of the $c$ cycle edges.
	However, after selecting the first cycle edge $\{u,v\}$ with origin $u$, the adjacency list $F[v]$ is non-empty and thus $v$ cannot be the origin of any edge selected in this phase.
	Therefore $F$ is acyclic and consists of $k \leq n-1$ edges which are all enqueued for later output.

	Assume for the extension phase, that pre-selected edges in $F$ have weight~0 and all other edges in $G$ have weight~1.
	Then, using the two queues and skipping edges to vertices already connected to the growing tree as shown in \Cref{alg:extension-prim} effectively mimics a priority queue over vertices with the priority of a vertex being the minimum weight of an edge that connects the vertex to the growing tree.
	Thus, correctness of Prim's algorithm implies that \Cref{alg:extension-prim} indeed produces a spanning tree $T$.
	As $F$ is a forest and can be extended to a spanning tree, every min-weight spanning tree (with weights defined as assumed above) must include all edges from $F$, which proves $F \subseteq T$.

	As the output finalization phase appends the missing edges $T \setminus F$ to the solution queue, each edge of the spanning tree $T$ is emitted exactly once by the enumeration algorithm.

	\subparagraph*{Time Complexity}

	We apply the accounting method to analyze the algorithm's time complexity.
	Each credit can be used to perform a constant number of steps.
	The initial balance is $\avgdeg > 0$, as we are aiming for a delay of $O(\avgdeg)$ without preprocessing time.

	Note that initially the algorithm does not know $\avgdeg$, so it cannot assume a higher lower bound on the delay than $\avgdeg \in \Omega(1)$.
	The first $\frac{n}{4}$ edges selected in the credit accumulation phase receive $6$ credits each.
	One credit pays for the selection itself, one pays for skipping the second edge endpoint in the iteration over $F$.
	Thus, selecting and emitting each edge takes amortized constant time.
	The remaining 4 credits per edge accumulate to $n$ credits that are used to compute $\avgdeg$.
	Now that the algorithm knows a proper lower bound on the delay, the next $\frac{n}{4}$ edges receive $2 + 8\avgdeg$ credits each.
	Again two credits pay for the selection and for skipping the second endpoint.
	The remaining $8\avgdeg$ credits per edge accumulate to $2m$ credits that pay for all computation in phases~2 and~3 which each run in $O(m)$ total time:

	In the extension phase, each undirected edge of the graph is appended to each queue at most twice (once per direction), so the loop in line~\ref{ln:prim-loop} of \Cref{alg:extension-prim} runs $O(m)$ times and the queue sizes cannot exceed $O(m)$.
	The \texttt{if}-block in line~\ref{ln:prim-visit} runs once per vertex, so the algorithm iterates over each vertex neighborhood at most twice, resulting in $O(m)$ steps.

	The output finalization phase iterates over $F$ once which takes $O(n)$ time.
	The additional edges appended to the solution queue in this phase receive no credit.

	As every edge appended to the solution queue receives credit in $O(\avgdeg)$ and all computation steps are paid for by credit accumulated before, the delay is in $O(\avgdeg)$ as claimed.

	\subparagraph*{Space Complexity}

	The credit accumulation phase uses $\Theta(n)$ lazy-initialized memory for $F$ and puts $O(n)$ edges into the solution queue.
	\Cref{alg:extension-prim} in the extension phase uses up to $\Theta(m)$ memory for the queue of all edges.
	Note however that, instead of expanding the neighborhood of vertex $v$ in lines~\ref{ln:prim-expand}f. into the two queues immediately, we can append only the vertex to the respective list.
	By expanding the vertex to the edges in its neighborhood lazily in line~\ref{ln:prim-dequeue}, we thereby shrink the length of both queues and thus the overall space complexity of \Cref{alg:extension-prim} to $O(n)$.
	The output finalization phase does not require additional memory on top of the provided $F$ and $T$ and the solution queue of at most $n-1$ edges.
\end{proof}

The analysis of the space complexity also shows why the extension phase does not simply contract all pre-selected edges to find the remaining tree via a depth first search:
Storing the contracted graph needs $\Theta(m)$ space in the worst case, while our Prim-based extension can be easily implemented to only use $\Theta(n)$ memory.

\subsection{Undirected Minimum Spanning Trees}
\label{sec:umst}

Any minimum weight spanning tree for an edge-weighted graph is also an unweighted spanning tree for the same graph.
Therefore the bound from \Cref{thm:uust-lower} carries over.
\begin{corollary}
	\label{cor:mst-avgdeg-lower}
	No algorithm can enumerate a minimum spanning tree of an edge-weighted connected graph with preprocessing in $o(m)$ and delay in $o(\avgdeg)$.
\end{corollary}

Given that no deterministic comparison-based MST algorithm with complexity in $O(m)$ is known, we can formulate a higher lower bound conditioned on the runtime of total-time MST algorithms, that generically applies to different machine models:
\begin{corollary}
	\label{cor:mst-algo-lower}
	Let $\totaltime$ be the runtime of an optimal MST algorithm.
	No algorithm can enumerate an MST of an edge-weighted connected graph with preprocessing in $o(\totaltime)$ and delay in $o(\frac{\totaltime}{n})$.
\end{corollary}

For smaller preprocessing time, we can use an adversary argument to give an unconditional and even higher lower bound on the delay:
As illustrated in \Cref{fig:mst-maxdeg-lower}, an adversary can force any algorithm to inspect either the degree of many vertices in a clique or the complete neighborhood of a single clique vertex before finding the first solution edge.

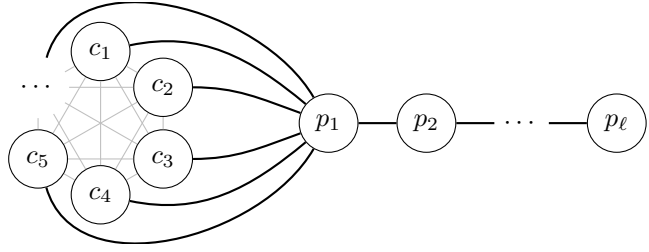
\begin{figure}[t]
	\centering
	\begin{tikzpicture}[vertex/.append style = {minimum width=2.2em}]
		\clip (-1.4, -1.6) rectangle (7.3, 1.6);
		\foreach \v in {1,...,5}{
			\node[vertex] (c\v) at (150-60*\v:0.95cm) {$c_\v$};
		}
		\node[vertex,draw=none] (c6) at (150:0.95cm) {$\cdots$};
		\foreach \u in {1,...,5}{
			\pgfmathtruncatemacro{\unext}{\u+1}
			\foreach \v in {\unext,...,6}{
				\draw[lightgray] (c\u) -- (c\v);
			}
		}
		\node[vertex] (p1) at (3cm,0) {$p_1$};
		\node[vertex, right=0.5 of p1] (p2) {$p_2$};
		\node[right=0.5 of p2] (p3) {$\cdots$};
		\node[vertex, right=0.5 of p3] (p4) {$p_\ell$};
		\draw[thick]
		(c2) edge[out=0, in=160] (p1)
		(c3) edge[out=0, in=-160] (p1)
		(c1) edge[out=12, in=140] (p1)
		(c4) edge[out=-12, in=-140] (p1)
		(c6) edge[out=74, in=120] (p1)
		(c5) edge[out=-74, in=-120] (p1)
		(p1) edge (p2)
		(p2) edge (p3)
		(p3) edge (p4)
		;
	\end{tikzpicture}
	\caption{The adversarial input for the proof of \Cref{thm:mst-maxdeg-lower} to trap an algorithm in the clique (many edges, no solution parts) before it can process the path (few edges, all are solution parts).}
	\label{fig:mst-maxdeg-lower}
\end{figure}

\begin{restatable}{theorem}{thmmstmaxdeglower}
	\label{thm:mst-maxdeg-lower}
	No algorithm can enumerate a minimum spanning tree of an edge-weighted connected graph with at least two different edge weights, maximum degree $\maxdeg$ and average degree $\avgdeg \in o(\maxdeg)$ with both preprocessing and delay in $o(\maxdeg)$.
\end{restatable}
\begin{proof}
	Consider an adversary according to \Cref{subsec:lower-techniques} that supplies an input graph with $n$ vertices to an algorithm.
	
	The adversary will, for arbitrary $k \in o(n) \cap \omega(1)$, construct a graph consisting of a clique of $k$  vertices $c_1, \ldots, c_k$, a path of $\ell = n - k$ vertices $p_1, \ldots, p_\ell$, and additional edges connecting all clique vertices to $p_1$ (see \Cref{fig:mst-maxdeg-lower}).
	All clique edges have edge weight $2$, all edges with at least one endpoint in the path have weight $1$.
	The unique MST of the graph therefore consists of all $n-1$ edges of weight $1$.
	Note that the graph has maximum degree $\maxdeg \in \Theta(k)$ and average degree $\avgdeg \in \Theta(\frac{k^2 + n}{n}) \subseteq o(k)$.

	While both the structure of the graph and all edge weights are fixed, the adversary is free to choose any consistent mapping of vertex IDs $\{1, \ldots, n\}$ to vertices in the graph prototype.
	Therefore the adversary can answer the first $k-3$ queries without revealing information on other vertices than $c_1, \ldots, c_{k-2}$.
	More specifically, it will answer each degree query with $k$ and each neighbor query with an edge of weight $2$ to one of those vertices.
	
	Now assume, the algorithm emits as solution part an edge $\{u,v\}$ before asking at least $k-2$ queries.
	Then, no matter what $u,v$ are, the adversary can choose the remaining vertex IDs in a way that $\{u,v\}$ is a clique edge of weight $2$ and thus not a correct solution part.
	
	Therefore, no algorithm can emit a solution part after $o(k) = o(\maxdeg)$ steps.
\end{proof}

Without preprocessing, we can match the combination of the delay lower bounds of \Cref{thm:mst-maxdeg-lower} and \Cref{cor:mst-algo-lower} with a corresponding upper bound;
assuming the upper bound on the total time of the referenced algorithm in \Cref{cor:mst-algo-lower} is computable.
Our enumeration algorithm essentially runs one iteration of Borůvka's algorithm \cite{boruvkaJistemProblemuMinimalnim1926} in the credit accumulation phase, which yields at least $\frac{n}{2}$ solution edges.
These initial solution edges give enough head start to, in the extension phase, run a total-time MST algorithm of our choice as black box to complete the tree.
As we have to make sure, that the black box algorithm produces a minimum spanning tree that includes the pre-selected edges, we set their weight to a new minimum weight (or answer edge weight comparisons accordingly in the comparison model).

\begin{restatable}{theorem}{thmmstmaxdegupper}
	\label{thm:mst-maxdeg-upper}
	Let $\totaltime$ be an upper bound on the runtime of an MST algorithm such that $\totaltime$ is computable in $O(\maxdeg n)$ steps and let $\totalspace$ be its space complexity.
	Enumerating the edges of an MST of an edge-weighted connected graph $G$ can be done without preprocessing, with delay in $O(\max(\maxdeg, \frac{\totaltime}{n}))$, and with space complexity in $\Theta(n + \totalspace)$.
\end{restatable}
\begin{proof}
	We describe the details of the enumeration algorithm in three phases as introduced in \Cref{subsec:upper-techniques}:
	
	\subparagraph*{Phase 1: Credit Accumulation}
	For each vertex $u$ in the order of increasing ID, the algorithm considers the edge $\{u,v\}$ with minimum edge weight; ties are broken in favor of the smallest vertex index $v$.
	As before, we call $u$ the \textit{origin} of the selected edge.
	If $u < v$, the edge is selected.
	Otherwise, for $u > v$, the algorithm only selects $\{u,v\}$, if $\{v,u\}$ is not already the selected edge for origin $v$; the algorithm decides this by recomputing the considered/selected edge for $v$.
	In addition to appending the selected edges to the enumeration algorithm's solution queue, they are also stored in a linked list $S$.
	At the end of the first phase, the algorithm transforms the list $S$ into an array $F$, setting $F[u] = v$ for each selected edge $\{u,v\}$ with origin $u$ and setting all other entries to \textsc{nil}.
	
	The first $\frac{n}{6}$ selected edges are emitted with delay in the order of the maximum degree of their two vertices.
	Simultaneously, the algorithm computes the maximum degree $\maxdeg$ of the whole graph.
	The next $\frac{n}{6}$ selected edges are emitted with delay in $\Theta(\maxdeg)$.
	Simultaneously, the algorithm computes the upper bound $\totaltime$ on the runtime of the MST algorithm used in the second phase.
	Additionally, the algorithm finds the minimum edge weight $W$ in the graph.
	All remaining edges are emitted with delay in $\Theta(\max(\maxdeg, \frac{\totaltime}{n}))$.
	
	\subparagraph*{Phase 2: Extension}
	In the second phase, all edges pre-selected in the credit accumulation phase are changed to have edge weight $W - 1$ (or, in the comparison-model, smaller than $W$).
	We call the graph with modified weights $G'$.
	
	The algorithm now runs an MST algorithm with runtime in $O(\totaltime)$ and space complexity in $O(\totalspace)$ to compute an MST for $G'$.
	We assume that the MST is produced in (or transformed with a DFS into) the form of a rooted tree stored in an array $T'$, where the arbitrary root $r$ has $T'[r] = \textsc{nil}$ and all other vertices $v$ store the edge $\{u,v\}$ that connects them to their parent $u$ as $T'[v] = u$.
	
	\subparagraph*{Phase 3: Output Finalization}
	Finally, the algorithm iterates over all edges in $F$ from the credit accumulation phase, sets the corresponding entries in the tree $T'$ from the extension phase to \textsc{nil} and enqueues for all remaining non-\textsc{nil}-entries $T'[v] = u$ the edge $\{u,v\}$ to the solution queue.
	
	\medskip
	
	We again first prove that the presented algorithm produces a valid spanning tree and emits each edge exactly once, before analyzing the algorithm's time and space complexity.
	
	\subparagraph*{Correctness}
	
	The edge selection in the credit accumulation phase is equivalent to the first iteration of modern formulations of Borůvka's algorithm with tie breaking \cite{boruvkaJistemProblemuMinimalnim1926,ericksonAlgorithms2019}.
	Thus, $k \leq n-1$ edges are selected and they form a forest $F$ that can be extended to an MST of~$G$.
	As each vertex is adjacent to at least one selected edge, $k \geq \frac{n}{2}$ edges are selected and enqueued for later output.
	
	The correctness of the used black-box MST algorithm in the extension phase implies that the tree $T'$ is a minimum spanning tree of $G'$.
	It remains to show that $F \subseteq T'$ and that $T'$ is also a minimum spanning tree of the original graph $G$.
	We formulate the proof here in terms of a new, fixed edge weight.
	Note that the arguments still hold in the comparison model, as we only require that the pre-selected edges compare as lighter when tested against any other edges.
	
	Assume $F \nsubseteq T'$, implying that there is at least one edge $e \in F$ with $e \notin T'$.
	Because $T'$ is a spanning tree, adding $e$ to $T'$ closes a cycle.
	As $F$ is a forest, there has to be at least one edge $e'$ in the cycle with $e' \notin F$.
	Let $w'(e)$ and $w'(e')$ be the weights of these two edges in $G'$.
	The weight of edge $e \in F$ was adjusted, thus we have $w'(e) = W - 1$, where $W$ is the smallest edge weight in $G$.
	Because $e' \notin F$, we have $w'(e') \geq W$ and therefore $w'(e') > w'(e)$.
	It follows that $(T' \cup \{e\}) \setminus \{e'\}$ is a spanning tree of $G'$ with smaller total weight than $T'$.
	But $T'$ is an MST of $G'$, therefore the initial assumption must be wrong and $F \subseteq T'$ holds.
	
	As the edge weight updates did not change the structure of the graph and $T'$ is a spanning tree of $G'$, it also is a spanning tree of the original graph $G$.
	Recall that the forest $F$ can be extended to an MST of $G$.
	Let $T$ be such an MST with weight $w(T)$ in $G$ and denote the weight of $T'$ in $G$ with $w(T')$.
	As $T$ is an MST in $G$, we have $w(T) \leq w(T')$.
	As $F\subseteq T \cap T'$, changing exactly the edge weights of $F$ in the construction of $G'$ reduces $w(T)$ and $w(T')$ by the same amount~$x$.
	The tree $T'$ is an MST in $G'$ which thus implies $w(T') - x \leq w(T) - x$.
	So, $w(T') = w(T)$ and $T'$ is an MST in $G$.
	
	The output finalization phase appends the missing edges $T' \setminus F$ to the solution queue.
	Therefore, each edge of the MST $T'$ is emitted exactly once by the enumeration algorithm.
	
	\subparagraph*{Time Complexity}
	
	As before we apply the accounting method to analyze the delay runtime.
	The initial balance is $\max(\maxdeg, \frac{\totaltime}{n})$.
	
	Each of the first $\frac{n}{6}$ selected edges $\{u,v\}$ receive $\deg(u) + \deg(v) + 6$ credits.
	This pays for the $O(\deg(u) + \deg(v))$ steps to consider and select the edge for origin $u$ and potentially consider and not select the edge for origin $v$ in a later iteration.
	(Note that not selecting a considered edge can only happen once per edge and only if the exact edge has been selected before, as the vertices are processed with increasing ID and the tie breaking decides towards the smaller ID.)
	The credits also pay for the $O(n)$ steps to compute the maximum degree $\maxdeg$.
	So after this step, the balance is again at least $\max(\maxdeg, \frac{\totaltime}{n})$.
	
	The next $\frac{n}{6}$ selected edges receive $(1 + 6 + 6)\maxdeg$ credits each.
	This pays for $O(\maxdeg)$ steps to consider and select the edge and, potentially, not selecting it again in a later iteration.
	It also pays for the $O(\maxdeg n)$ steps to compute the upper bound $\totaltime$ on the MST algorithm used in the extension phase and for the $O(m)$ steps to find the minimum edge weight $W$.
	After this step, the balance is again at least $\max(\maxdeg, \frac{\totaltime}{n})$.
	
	All remaining at least $\frac{n}{6}$ selected edges receive $\maxdeg + 6\frac{\totaltime}{n} + 12$ credits.
	This pays for $O(\maxdeg)$ steps to consider and select the edge and, potentially, not select it again in a later iteration.
	After this step, the balance is at least $\totaltime + 2n$ credits.
	
	As last step in the credit accumulation phase, $n$ of those credits pay for transforming the linked list $S$ of selected edges to the array $F$, leaving $\totaltime + n$ credits that remain for the next two phases.
	
	The extension phase runs in $O(\totaltime)$; including the time needed to update edge weights:
	Observe that the algorithm can check in constant time for any edge $\{u,v\}$ whether it was selected in phase~1 by testing $F[u] = v$ and $F[v] = u$.
	Thus, the algorithm can either modify the input $G$ by iterating over all edges once and replacing the selected edges' weights, or proxy every edge weight comparison of the following operations to have the pre-selected edges compare as lighter to the remaining ones.
	
	In the output finalization phase, the algorithm iterates over $F$ once in $O(n)$ time.
	Additional edges appended to the solution queue in this phase receive no credit.
	
	All edges receive at most $13\maxdeg + 6\frac{\totaltime}{n} + 12$ credits, and never more than is warranted by the current knowledge of the algorithm about $\maxdeg$ and $\totaltime$.
	Thus the delay is in $O(\max(\maxdeg, \frac{\totaltime}{n}))$.
	
	\subparagraph*{Space Complexity}
	
	In the credit accumulation phase, both the linked list $S$ as well as the array $F$ need $\Theta(n)$ memory.
	Note that having $S$ and building the array $F$ at the end of the phase removes any need for lazy-initialized memory.
	The black-box MST algorithm requires $O(\totalspace)$ space.
	The output finalization phase does not require additional memory.
	Therefore, the overall memory consumption is in $O(n + \totalspace)$.
\end{proof}

For graphs with large $\maxdeg$ one can trade in $\Theta(n)$ preprocessing time to match the smaller delay lower bound of $O(\frac{\totaltime}{n})$:
The algorithm goes through the vertices in order of increasing degree to build enough head start before dealing with high-degree vertices later.

\begin{restatable}{theorem}{thmmstalgoupper}
	\label{thm:mst-algo-upper}
	Let $\totaltime$ be an upper bound on the runtime of an MST algorithm such that $\totaltime$ is computable in $O(\avgdeg n)$ steps and let $\totalspace$ be its space complexity.
	Enumerating the edges of an MST of an edge-weighted connected graph $G$ can be done with preprocessing in $\Theta(n)$, with delay in $O(\frac{\totaltime}{n})$, and with space complexity in $\Theta(n + \totalspace)$.
\end{restatable}
\begin{proof}
	During preprocessing, our algorithm computes the average degree $\avgdeg$.
	It then sorts the vertices by ascending vertex degree by creating an array of $n$ buckets, one for each possible degree.
	After putting each vertex in the according bucket, iterating over the array and its buckets once gives the sorted sequence in $\Theta(n)$ time.
	
	The algorithm then implicitly renames the vertices according to this new order, by building two translation arrays $A, B$ of length $n$ such that for a vertex with old ID $v$ and new ID $v'$ it holds that $A[v] = v'$ and $B[v'] = v$.
	Using $A, B$ as indirection, all following computation steps can access $G$ as if the vertices were renamed and all edge references were updated accordingly.
	
	The rest of the algorithm is identical to the algorithm of \Cref{thm:mst-maxdeg-upper}, except for the output delay:
	The first $\frac{n}{3}$ edges are emitted with delay in $\Theta(\avgdeg)$ and receive $\avgdeg$ credits.
	All remaining edges receive $\Theta(\frac{\totaltime}{n})$ credits to be emitted with delay in $\Theta(\max(\avgdeg, \frac{\totaltime}{n})) = \Theta(\frac{\totaltime}{n})$.
	
	With the vertices reordered by increasing degree it holds that for all $1 \leq x \leq n$ we have
	\(\frac{1}{x} \cdot \sum_{v = 1}^{x} \deg(v) \leq \avgdeg\).
	Therefore, emitting with delay in the order of the average degree is indeed possible in the first phase without running out of solution parts / credits.
\end{proof}

\Cref{fig:mst-summary} summarizes the tradeoff between preprocessing time and delay for the enumeration of edges of a minimum spanning tree for an edge-weighted graph.
The figure highlights a gap for preprocessing in $\Omega(\maxdeg) \cap o(n)$ and delay in $\Omega(\frac{\totaltime}{n}) \cap o(\maxdeg + \frac{\totaltime}{n})$, where we do not yet know whether enumeration is possible or not.

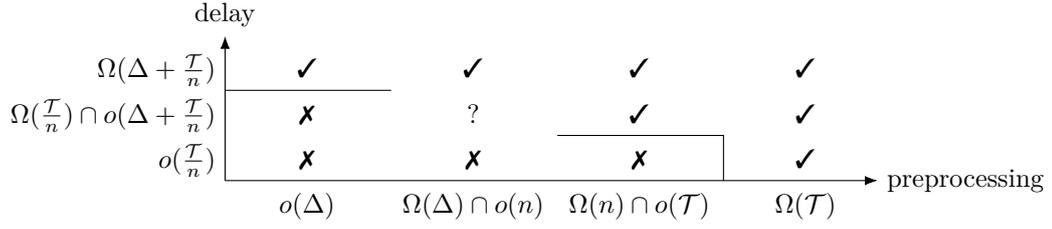
\begin{figure*}[t]
	\centering
	\begin{tikzpicture}[
			cell/.style={minimum width=2.2cm, minimum height=0.6cm},
			column 1/.style={anchor=base east, nodes={minimum width=0cm}},
		]

		\node [matrix of nodes, ampersand replacement=\&, nodes in empty cells, nodes={cell},column sep=-\pgflinewidth, row sep=-\pgflinewidth,text depth=0.5ex,text height=2ex] (m) {
			$\Omega(\maxdeg + \frac{\totaltime}{n})$ \& \cmark \& \cmark \& \cmark \& \cmark \\
			$\Omega(\frac{\totaltime}{n}) \cap o(\maxdeg + \frac{\totaltime}{n})$ \& \xmark \& ? \& \cmark \& \cmark \\
			$o(\frac{\totaltime}{n})$ \& \xmark \& \xmark \& \xmark \& \cmark \\
			\& $o(\maxdeg)$ \& $\Omega(\maxdeg) \cap o(n)$ \& $\Omega(n) \cap o(\totaltime)$ \& $\Omega(\totaltime)$ \\
		};
		
		\draw
			(m-4-4.north east) -- (m-3-4.north east) -- (m-3-3.north east)
			(m-2-1.north east) -- (m-2-2.north east)
		;

		\node (rightlabel) at ($(m-4-5.north east)+(1,0)$) {preprocessing};
		\draw (m-4-1.north east) edge[->,>=Latex] (rightlabel);
		\node (uplabel) at ($(m-1-1.north east)+(0,0.4)$) {delay};
		\draw (m-4-1.north east) edge[->,>=Latex] (uplabel);
	\end{tikzpicture}
	\caption{Tradeoff between preprocessing and delay for MST enumeration with a black-box MST algorithm with total time~$\totaltime$.}
	\label{fig:mst-summary}
\end{figure*}

\subsection{Directed Unweighted Spanning Trees}
\label{sec:dust}

We analyze two variants of computing a directed spanning tree (DST) for a graph that is guaranteed to have one:
First, if a suitable root vertex $r$ for such a DST is given, second without given root.

Given a root $r$, the search tree of a depth first search started in $r$ is a directed spanning tree \cite{cormenIntroductionAlgorithms2022}.

\begin{corollary}
	\label{cor:dst-with-r-upper}
	Given an unweighted directed graph $G$ that has a directed spanning tree rooted in a given vertex $r$, such a DST can be computed in $O(m)$ by a DFS.
\end{corollary}

Without given root this runtime is still achievable, as one can find a suitable root in $O(m)$:
Build the acyclic graph $G^\textrm{SCC}$ in which each strongly connected component is contracted to a single node~\cite{cormenIntroductionAlgorithms2022}.
Pick any vertex in the unique source of $G^\textrm{SCC}$ as root.

\begin{restatable}{theorem}{thmdstupper}
	\label{thm:dst-upper}
	Given an unweighted directed graph $G$ that has a directed spanning tree, such a DST can be computed in $O(m)$.
\end{restatable}
\begin{proof}
	The following steps compute a DST and  have a runtime in $O(m)$:
	\begin{enumerate}
		\item Find all strongly connected components (SCCs) of $G$ with \cite{tarjanDepthFirstSearchLinear1972,cormenIntroductionAlgorithms2022}.
		\item Construct the acyclic component graph $G^\textrm{SCC}$, in which each SCC is contracted to a single vertex \cite{cormenIntroductionAlgorithms2022}.
		\item Find the unique source $S$ in $G^\textrm{SCC}$.
		\item Chose any vertex $r \in S$ and return the DST produced by a DFS on $G$ started on~$r$.
	\end{enumerate}
	
	We only need to argue the correctness and runtime of steps~3 and~4: First, assume towards contradiction that $G^\textrm{SCC}$ has two sources $S$ and $S'$.
	By definition of SCCs and sources, a vertex $u \in S$ and a vertex $v \in S'$ are not reachable in $G$ from any vertex outside their own SCC.
	Thus, there is no common root $r$ with a path to both $u$ and $v$ and therefore there is no directed spanning tree in $G$.
	
	Thus step 3 correctly finds a  unique source $S$ in  $G^\textrm{SCC}$.
	Since $S$ is a strongly connected component, there is a path from every vertex $r \in S$ to every other vertex $x \in S$.
	As there is no other source but $S$ in the component graph, there is a path from $S$ to all other SCCs in $G^\textrm{SCC}$ -- otherwise $G$ would not be disconnected, again contradicting the existence of a directed spanning tree. 
	It follows that there is also a path from every vertex $r \in S$ to every other vertex $x \in G$.
	So, any $r \in S$ is a valid source for a DST in $G$ and by \Cref{cor:dst-with-r-upper} step~4 is correct and runs in $O(m)$.
	
	Finding $S$ requires at most $O(m)$ steps by transposing the component graph and looking for vertices without outgoing edges.
\end{proof}

The lower bound from \Cref{thm:uust-lower} for undirected graphs extends to the directed case:
\begin{corollary}
	\label{cor:dst-lower}
	Given is an unweighted directed graph $G$ that has a DST.
	No algorithm can compute a DST in $o(m)$.
	No algorithm can enumerate a DST with preprocessing in $o(m)$ and delay in $o(\avgoutdeg)$.
\end{corollary}

For the total-time variant this is already a tight bound.
We can, however, prove a stronger lower bound for the enumeration of directed spanning trees without given root, that rules out any meaningful enumeration even for dense graphs.
The core idea of the proof is that any vertex in a DST is the endpoint of at most one edge.
Were an algorithm to emit an edge $(u,v)$ early, the adversary could point key edges (in \Cref{fig:dste-lower} these are crossing edges between two cliques), to vertex $v$; and one of these key edges has to be included in any DST, making $(u,v)$ wrong.

\begin{restatable}{theorem}{thmdstelower}
	\label{thm:dste-lower}
	Given is an unweighted directed graph that has a DST.
	Without given root, no algorithm can enumerate such a DST with both preprocessing and delay in $o(m)$.
\end{restatable}
\begin{proof}
	Consider a graph  consisting of two bi-directed cliques with sizes $\lfloor \frac{n}{2} \rfloor$ and $\lceil \frac{n}{2} \rceil$, where one of the cliques is missing the edges between two vertices $s$ and $t$, and with two additional edges from $s$ and $t$ to a vertex $v$ in the other clique (see \Cref{fig:dste-lower}).
	
	Assume now, an algorithm emits an edge $(u,v)$ after at most $(\lfloor \frac{n}{2} \rfloor-2)(\lfloor \frac{n}{2} \rfloor-1)$ queries.
	At this point in time, there are at least two vertices per clique for which the algorithm has not queried all outgoing edges.
	Therefore, until this point, the adversary is able to only reveal edges within the cliques.
	So, the vertices $u$ and $v$ must belong to the same clique, or the adversary could place the crossing edges such that $(u,v)$ does not exist at all.
	Additionally the adversary can ensure that in each clique there are at least two vertices that are not yet connected by a queried edge in either direction.
	Let $r,s$ be such two vertices in the clique not containing $u,v$.
	Note that within each clique, all vertices have the same out-degree, so the algorithm cannot identify $s$ and $t$ using out-degree queries.
	
	Now, the adversary can place two edges $(r,v)$ and $(s,v)$.
	Any directed spanning tree of the graph must use exactly one of them.
	However, together with the emitted edge $(u,v)$, vertex $v$ would then have two incoming edges in the algorithm's output, contradicting the definition of a directed spanning tree.
	This shows, that the emitted solution part is wrong.
\end{proof}
\begin{figure}[t]
	\centering
	\begin{tikzpicture}[vertex/.append style = {minimum width=1.8em}]
		\foreach \v/\l in {1/,2/$v$,3/$u$,4/}{
			\node[vertex] (a\v) at (162-72*\v:1.15cm) {\l};
		}
		\node[vertex,draw=none] (a5) at (162:1.15cm) {$\cdots$};
		\foreach \u in {1,...,4}{
			\pgfmathtruncatemacro{\unext}{\u+1}
			\foreach \v in {\unext,...,5}{
				\draw[lightgray] (a\u) edge[directed,bend left=5] (a\v);
				\draw[lightgray] (a\v) edge[directed,bend left=5] (a\u);
			}
		}

		\begin{scope}[xshift=5cm]
			\foreach \v/\l in {1/,3/,4/$t$,5/$s$}{
				\node[vertex] (b\v) at (162-72*\v:1.15cm) {\l};
			}
			\node[vertex,draw=none] (b2) at (18:1.15cm) {$\cdots$};
			\foreach \u in {1,...,3}{
				\pgfmathtruncatemacro{\unext}{\u+1}
				\foreach \v in {\unext,...,4}{
					\draw[lightgray] (b\u) edge[directed,bend left=5] (b\v);
					\draw[lightgray] (b\v) edge[directed,bend left=5] (b\u);
				}
			}
			\foreach \v in {1,...,3}{
				\draw[lightgray] (b5) edge[directed,bend left=5] (b\v);
				\draw[lightgray] (b\v) edge[directed,bend left=5] (b5);
			}
			\draw (b5) edge[directed] (a2);
			\draw (b4) edge[directed] (a2);
		\end{scope}
	\end{tikzpicture}
	\caption{The prototype graph for the adversary for the proof of \Cref{thm:dste-lower}. A DST enumeration algorithm cannot emit an edge $(u,v)$ before having found the crossing edges.}
	\label{fig:dste-lower}
\end{figure}
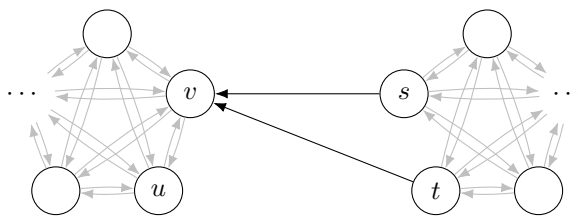

\Cref{cor:dst-with-r-upper} and \Cref{thm:dst-upper} show that providing a DST root or not does not change the time complexity in the total-time model.
This is different in the enumeration variant, as with given root, one \textit{can} enumerate the edges of a DST with  delay in $O(\maxoutdeg)$ by adjusting the BFS variation given in the proof of Theorem~1 in~\cite{caselShortestDistancesEnumeration2024}.
Starting this shortest distance enumeration with root $s$, we also store and then emit the edges used to reach a vertex on a shortest path from $s$. This adjustment outputs a shortest distances tree, thus in particular a DST, and immediately yields the following result.
\begin{corollary}
	\label{cor:dste-with-r-upper}
	Given an unweighted directed graph $G$ that has a directed spanning tree rooted in a given vertex $r$, such a DST can be enumerated with delay in $O(\maxoutdeg)$, without preprocessing, with $\Theta(n)$ lazy-initialized memory and space complexity in $\Theta(n)$.
\end{corollary}

Note that this leaves a gap to the $\Omega(\avgoutdeg)$ delay lower bound from \Cref{cor:dst-lower}.
While the $\Omega(\maxdeg)$ delay lower bound for single source shortest distances from \cite{caselShortestDistancesEnumeration2024} does not transfer to the DST problem, we can use a similar technique to also show a matching lower bound for directed spanning tree enumeration:
An adversary can hide the entrance to a bi-directed path behind many edges in a clique.
For suitable clique- and path sizes this forces an algorithm to essentially fully explore the clique with many edges and few solution parts before being able to emit one of the many solution parts in the path.

\begin{theorem}
	\label{thm:dste-with-r-lower}
	Given is an unweighted directed graph $G$ that has a directed spanning tree rooted in a given vertex $r$.
	No algorithm can enumerate such a DST with preprocessing in $o((\maxoutdeg)^2)$ and delay in $o(\maxoutdeg)$, even if $\avgoutdeg \in o(\maxoutdeg)$.
\end{theorem}
\begin{proof}
	We use  the adversarial setup from \Cref{subsec:lower-techniques} with $k$ to be chosen later.
	Consider a graph consisting of a bi-directed clique of $k$ vertices and a bi-directed path of $n-k$ vertices.
	The given root $r$ lies in the clique.
	One other clique vertex $c$ has no edge to $r$ and instead an edge to one of the two end-vertices $p$ or $q$ of the path.
	(See \Cref{fig:dste-with-r-lower} for an illustration.)
	Note that this graph has maximum degree $\maxoutdeg=k-1$.

	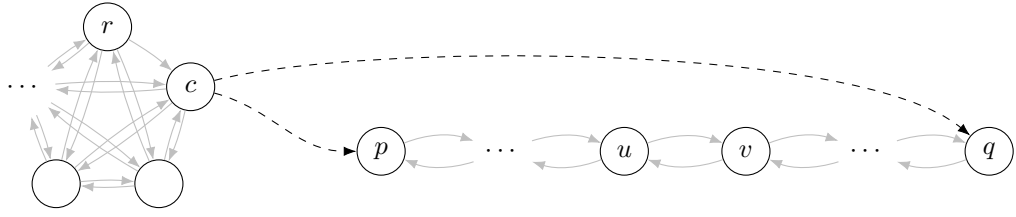
\begin{figure*}[t]
		\centering
		\begin{tikzpicture}[vertex/.append style = {minimum width=1.8em}]
			\foreach \v/\l in {1/$r$,3/,4/,5/$c$}{
				\node[vertex] (c\v) at (18+72*\v:1.15cm) {\l};
			}
			\node[vertex,draw=none] (c2) at (162:1.15cm) {$\cdots$};
			\foreach \u in {1,...,3}{
				\pgfmathtruncatemacro{\unext}{\u+1}
				\foreach \v in {\unext,...,4}{
					\draw[lightgray] (c\u) edge[directed,bend left=5] (c\v);
					\draw[lightgray] (c\v) edge[directed,bend left=5] (c\u);
				}
			}
			\foreach \v in {2,...,4}{
				\draw[lightgray] (c5) edge[directed,bend left=5] (c\v);
				\draw[lightgray] (c\v) edge[directed,bend left=5] (c5);
			}
			\draw[lightgray] (c1) edge[directed,bend left=5] (c5);
			
			\begin{scope}[xshift=2cm,yshift=-0.5cm,xscale=1.6]
				\foreach \l/\v in {$p$/1,$u$/3,$v$/4,$q$/6}{
					\node[vertex] (p\v) at (\v,0) {\l};
				}
				\foreach \v in {2,5}{
					\node[vertex,draw=none] (p\v) at (\v,0) {$\cdots$};
				}
				\foreach \i [evaluate=\i as \j using int(\i+1)] in {1, ..., 5}{
					\draw[lightgray] (p\i) edge[directed,bend left] (p\j);
					\draw[lightgray] (p\j) edge[directed,bend left] (p\i);
				}
			\end{scope}
			\draw (c5) edge[directed, out=-15, in=180, dashed] (p1);
			\draw (c5) edge[directed, out=15, in=140, looseness=0.5, dashed] (p6);
		\end{tikzpicture}
		\caption{The adversarial graph for \Cref{thm:dste-with-r-lower}. A DST enumeration algorithm cannot emit an edge from the bi-directed path before having found the outgoing edge from the clique.}
		\label{fig:dste-with-r-lower}
	\end{figure*}

	The adversary places each edge to $r$ and also the edge from $c$ to $p$ or $q$ last in the respective adjacency lists.
	Further, observe that for any degree query in the clique, the adversary answers $k-1$, so no algorithm can identify $c$ with degree queries.

	Assume, after less than $(k-1)^2$ queries to the adversary, an algorithm emits an edge $(u,v)$ among two path vertices or an edge that connects the clique to the path.
	Even if all the algorithm's queries are adjacency queries on clique vertices, the adversary can up to this point only answer with edges within the clique.
	Further, there is at least one vertex in the clique where the algorithm has not seen the whole neighborhood.
	The adversary chooses this vertex as $c$ and connects it either to $p$ or $q$, picking the endpoint of the path that makes the emitted edge wrong.

	Thus, any algorithm has to make at least $(k-1)^2$ queries before emitting a solution edge outside of the clique.
	As there can only be $k-1$ tree edges inside the clique, the adversarial setup shows a lower bound on emitting the $k$-th edge of $\Omega(\prefixtime)=\Omega(k^2)$.
	With $\maxoutdeg = k-1$, this shows a lower bound of $\Omega(k^2) = \Omega((\maxoutdeg)^2)$ for preprocessing or $\Omega(k) = \Omega(\maxoutdeg)$ for delay.

	It remains to pick $k$ to show the desired low average degree.
	The graph consists of $k \cdot (k-1) + 2 \cdot (n-k-1)$ edges, so for $k \in o(n)$ the average degree is $\avgoutdeg \in o(\maxoutdeg)$.
\end{proof}

\subsection{Directed Minimum Spanning Trees}
\label{sec:dmst}

Efficient implementations of \textit{Edmonds' algorithm} need little more than linear time to solve any instance completely \cite{gabowEfficientAlgorithmsFinding1986,mendelsonMeldingPriorityQueues2006}.
A simple adversary construction rules out essentially any meaningful enumeration, as any algorithm has to inspect almost all of the input graph as preprocessing:
Before emitting a first edge $(u,v)$ as part of an MDST for a bi-directed clique, an algorithm has to query $\Theta(m)$ adjacencies to find all edges incident to $v$ to make sure that the single spanning tree edge incident to $v$ is one of minimum weight.

\begin{restatable}{theorem}{thmmdstelower}
	\label{thm:mdste-lower}
	Minimum directed spanning tree enumeration cannot be solved with both preprocessing and delay in $o(m)$.
\end{restatable}
\begin{proof}
	We again use the adversarial setup from \Cref{subsec:lower-techniques}.
	
	Consider a bi-directed clique of $n$ vertices, split into two groups of $\lceil \frac{n}{2} \rceil$ and $\lfloor \frac{n}{2} \rfloor$ vertices, respectively.
	All adjacency lists first specify all edges to vertices in the same group, followed by the edges to vertices of the other group.
	(See \Cref{fig:mdste-lower} for an illustration.)
	For the first $(\lfloor \frac{n}{2} \rfloor)^2$ edge queries, the adversary consistently reports an edge weight of 2.
	Let $v$ be any vertex from the first group and observe, that to find all incident edges $(w,v)$, any algorithm needs to ask at least $\lceil \frac{n}{2} \rceil-1$ neighbor queries to vertices from the first group and $(\lfloor \frac{n}{2} \rfloor)^2$ neighbor queries to vertices from the second group.
	(Similarly for any vertex $v$ from the second group.)
	Further, any degree query always returns the answer $n-1$.
	
	Assume now, an algorithm  emits  an edge $(u,v)$ as solution part after only $(\lfloor \frac{n}{2} \rfloor)^2$ steps.
	Then there is at least one edge $(w,v)$ with $w \neq u$ that the algorithm has not yet queried.
	The adversary sets the weight of this edge to 1 and all remaining edge weights to 2.
	Thus, any correct MDST of the graph must use $(w,v)$, and as $v$ must not be the endpoint of two edges, the solution part $(u,v)$ is wrong.
	Therefore, an algorithm has to make at least $\Omega((\lfloor \frac{n}{2} \rfloor)^2) = \Omega(m)$ queries before emitting the first solution part.
	
	Note that this proof applies to both the setting with and without given root vertex.
\end{proof}

\begin{figure}[t]
	\centering
	\begin{tikzpicture}[vertex/.append style = {minimum width=2em}]
		\clip (-1.4, -1.2) rectangle (3.4, 1.6);
		
		\node[vertex] (a1) at (110:1cm) {$v$};
		\node[vertex] (a2) at (180:1cm) {};
		\node[vertex,draw=none] (a3) at (250:1cm) {$\cdots$};
		\foreach \u in {1,2}{
			\pgfmathtruncatemacro{\unext}{\u+1}
			\foreach \v in {\unext,...,3}{
				\draw[lightgray] (a\u) edge[directed,bend left=5] (a\v);
				\draw[lightgray] (a\v) edge[directed,bend left=5] (a\u);
			}
		}
		\begin{scope}[xshift=2cm]
			\node[vertex] (b1) at (70:1cm) {$w$};
			\node[vertex] (b2) at (0:1cm) {};
			\node[vertex,draw=none] (b3) at (290:1cm) {$\cdots$};
			\foreach \u in {1,2}{
				\pgfmathtruncatemacro{\unext}{\u+1}
				\foreach \v in {\unext,...,3}{
					\draw[lightgray] (b\u) edge[directed,bend left=5] (b\v);
					\draw[lightgray] (b\v) edge[directed,bend left=5] (b\u);
				}
			}
		\end{scope}
		
		\foreach \u in {1,...,3}{
			\foreach \v in {1,...,3}{
				\ifthenelse{\u = 1 \and \v = 1}{
					\draw[dashed] (b\v) edge[directed,bend right=10] node[above] {$1$} (a\u);
				}{
					\draw[lightgray, dashed] (b\v) edge[directed,bend right=5] (a\u);
				}
				\draw[lightgray, dashed] (a\u) edge[directed,bend right=5] (b\v);
			}
		}
	\end{tikzpicture}
	\caption{The prototype graph for the adversary for the proof of \Cref{thm:mdste-lower}, where all gray edges have weight $2$. The adversary reveals the dashed edges as late as possible such that an algorithm has to query many edges before finding all edges incident to $v$; any of which could be the single correct edge to $v$ with weight $1$ in all MDSTs.}
	\label{fig:mdste-lower}
\end{figure}
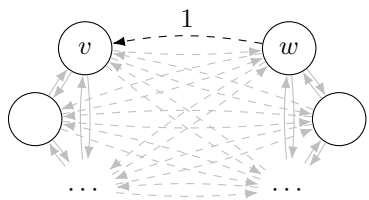

\section{Experimental Evaluation}
\label{sec:experiments}

We implemented MST algorithms in Rust and executed the experiments on a compute server with 256 GB RAM and an Intel Xeon Silver 4314 CPU with 2.40 GHz.
More specifically, we compared the total time MST algorithms by Prim \cite{primShortestConnectionNetworks1957} (with a binary heap \cite{sekineBinary_heap_plus2022}), Kruskal \cite{kruskalShortestSpanningSubtree1956} (with union-find with union-by-rank \cite{cormenIntroductionAlgorithms2022} and path-halfing \cite{vanderweideDatastructuresAxiomaticApproach1980}), and Boruvka \cite{boruvkaJistemProblemuMinimalnim1926}, with our enumeration approach of \Cref{thm:mst-algo-upper} with each of the other three algorithms as black box.
Additionally, as Prim's algorithm already fixes edges one at a time, we re-interpreted it as enumeration algorithm; meaning that whenever the algorithm fixes an MST edge, the algorithm is interrupted and the edge is immediately emitted as solution part.
Note that for the enumeration variants we did not implement holding back solution parts in a queue, but instead emitted each solution edge as soon as it was fixed.

For the evaluation we generated random graphs in the $G(n,p)$ model \cite{gilbertRandomGraphs1959} with the number $n$ of vertices ranging from $100$ to $200\,000$ and an edge probability $p$ between $n^{-0.75}$ and $0.25$.
For each input size we ran the algorithms on 10~random instances with 5~runs each and took the average time in nanoseconds for, among others, three values:
\begin{itemize}
	\item The time until the algorithm produced a first solution part;
	\item the maximum so-called \textit{incremental} delay, meaning the maximum of $\frac{\text{elapsed time}}{\text{number of emitted parts}}$;
	\item the total time of the run.
\end{itemize}
Note that while inspecting the worst case delay gives stronger bounds in the theoretical analysis, looking at incremental delay is the more reasonable measurement in practice.
This avoids computational overhead for holding back solution parts and still guarantees the availability of the $k$th part after $k \cdot \textit{incremental delay}$ time.

For the generated instances, Prim's total-time algorithm and its enumeration variant consistently performed best among their respective algorithm categories.
\Cref{fig:prim-experiments} shows the direct comparison of (the re-interpreted) Prim's algorithm with the enumeration approach of \Cref{thm:mst-algo-upper} with Prim's algorithm as black box.
As expected, the additional effort for enumerating solution parts early leads to a larger total time.
However, the enumeration algorithm produces the first solution part one order of magnitude faster and also has the benefit of a smaller incremental delay.
This confirms that the theoretical advantage of the enumeration variant transfers to practice if the time to a first solution part and/or the delay are of more importance than the total time.

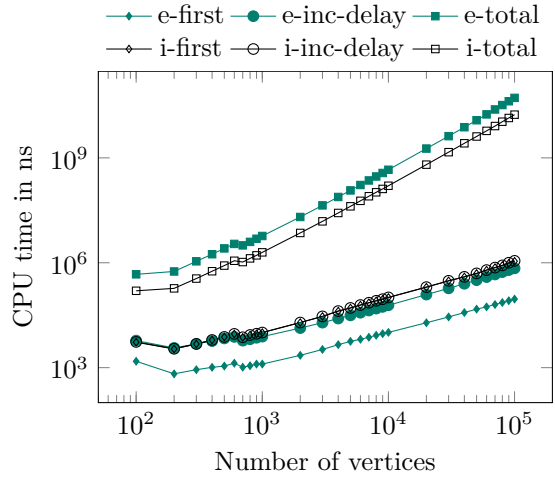
\begin{figure}[t]
	\centering
	\begin{tikzpicture}
		\begin{loglogaxis}[
			width=215pt,
			height=170pt,
			xlabel={Number of vertices},
			ylabel={CPU time in ns},
			ymin=100,
			ymax=3e11,
			legend entries={
				e-first,
				e-inc-delay,
				e-total,
				i-first,
				i-inc-delay,
				i-total,
			},
			legend style={at={(0.5,1)},anchor=south,draw=none},
			legend columns=3,
		]
			\addplot[plotcolor, mark=diamond*, mark size=1.4] table
				[x=size,y=first_output_avg, col sep=comma] {./experiment_data/aggregated_0.25.enum-prim.csv};
			\addplot[plotcolor, mark=*] table
				[x=size,y=delay_inc_max_avg, col sep=comma] {./experiment_data/aggregated_0.25.enum-prim.csv};
			\addplot[plotcolor, mark=square*, mark size=1.4] table
				[x=size,y=total_time_avg, col sep=comma] {./experiment_data/aggregated_0.25.enum-prim.csv};

			\addplot[black, mark=diamond, mark size=1.4] table
				[x=size,y=first_output_avg, col sep=comma] {./experiment_data/aggregated_0.25.incremental-prim.csv};
			\addplot[black, mark=o] table
				[x=size,y=delay_inc_max_avg, col sep=comma] {./experiment_data/aggregated_0.25.incremental-prim.csv};
			\addplot[black, mark=square, mark size=1.4] table
				[x=size,y=total_time_avg, col sep=comma] {./experiment_data/aggregated_0.25.incremental-prim.csv};

			\addplot[draw=none, name path=e-first-min] table
				[x=size,y=first_output_lower_quartile, col sep=comma] {./experiment_data/aggregated_0.25.enum-prim.csv};
			\addplot[draw=none, name path=e-first-max] table
				[x=size,y=first_output_upper_quartile, col sep=comma] {./experiment_data/aggregated_0.25.enum-prim.csv};

			\addplot[draw=none, name path=e-inc-delay-min] table
				[x=size,y=delay_inc_max_lower_quartile, col sep=comma] {./experiment_data/aggregated_0.25.enum-prim.csv};
			\addplot[draw=none, name path=e-inc-delay-max] table
				[x=size,y=delay_inc_max_upper_quartile, col sep=comma] {./experiment_data/aggregated_0.25.enum-prim.csv};

			\addplot[draw=none, name path=e-total-min] table
				[x=size,y=total_time_lower_quartile, col sep=comma] {./experiment_data/aggregated_0.25.enum-prim.csv};
			\addplot[draw=none, name path=e-total-max] table
				[x=size,y=total_time_upper_quartile, col sep=comma] {./experiment_data/aggregated_0.25.enum-prim.csv};

			\addplot[draw=none, name path=i-first-min] table
				[x=size,y=first_output_lower_quartile, col sep=comma] {./experiment_data/aggregated_0.25.incremental-prim.csv};
			\addplot[draw=none, name path=i-first-max] table
				[x=size,y=first_output_upper_quartile, col sep=comma] {./experiment_data/aggregated_0.25.incremental-prim.csv};

			\addplot[draw=none, name path=i-inc-delay-min] table
				[x=size,y=delay_inc_max_lower_quartile, col sep=comma] {./experiment_data/aggregated_0.25.incremental-prim.csv};
			\addplot[draw=none, name path=i-inc-delay-max] table
				[x=size,y=delay_inc_max_upper_quartile, col sep=comma] {./experiment_data/aggregated_0.25.incremental-prim.csv};

			\addplot[draw=none, name path=i-total-min] table
				[x=size,y=total_time_lower_quartile, col sep=comma] {./experiment_data/aggregated_0.25.incremental-prim.csv};
			\addplot[draw=none, name path=i-total-max] table
				[x=size,y=total_time_upper_quartile, col sep=comma] {./experiment_data/aggregated_0.25.incremental-prim.csv};
		\end{loglogaxis}
	\end{tikzpicture}
	\caption{Time to first output, incremental delay and total time for (e) MST enumeration with Prim's algorithm as black box and (i) MST computation with Prim's algorithm, interrupted at each fixed solution part.}
	\label{fig:prim-experiments}
\end{figure}

\section{Conclusions and Future Work}
\label{sec:future-work}

In this work we proved upper and lower bounds on the required preprocessing and delay for enumerating the edges of (minimum) (directed) spanning trees.
While several of our theoretical results are already tight, there still are gaps to be closed (cf. \Cref{fig:mst-summary}).

An important next step is to incorporate more real world requirements in the model:
Most networks admit for many different MSTs, and for some problems, not all of them work equally well~\cite{agmonGivingTreeConstructing2008}.
Also, the order in which tree edges are emitted can be relevant to efficiently interleave MST computation with following processing steps.
Thus, to bring our theoretical results closer to practical application, additional requirements on the produced output need to be investigated with regard to the possibility of enumerating solution parts with little preprocessing and delay.


\bibliography{spanning-tree-enumeration}

\end{document}